\documentclass[12pt]{article}%
\usepackage{amsmath}
\usepackage{amsfonts}
\usepackage{amssymb}
\usepackage{graphicx}
\usepackage{cite}%
\providecommand{\U}[1]{\protect\rule{.1in}{.1in}}

\providecommand{\U}[1]{\protect \rule{.1in}{.1in}}
\providecommand{\U}[1]{\protect \rule{.1in}{.1in}}
\newtheorem{theorem}{Theorem}

\newtheorem{proposition}[theorem]{Proposition}

\newenvironment{proof}[1][Proof]{\noindent \textbf{#1.} }{\  \rule{0.5em}{0.5em}}
\begin{document}

\title{Coupled quintessence with double exponential potentials}
\author{Koralia Tzanni\thanks{\texttt{tzanni@aegean.gr}} and John
Miritzis\thanks{\texttt{imyr@aegean.gr}},\\Department of Marine Sciences, University of the Aegean,\\University Hill, Mytilene 81100, Greece}
\maketitle

\begin{abstract}
We study flat Friedmann-Robertson-Walker (FRW) models with a perfect fluid
matter source and a scalar field non minimally coupled to matter having a
double exponential potential. It is shown that the scalar field almost always
diverges to infinity. Under conditions on the parameter space, we show that
the model is able to give an acceptable cosmological history of our universe,
that is, a transient matter era followed by an accelerating future attractor.
It is found that only a very weak coupling can lead to viable cosmology. We
study in the Einstein frame, the cosmological viability of the asymptotic form
of a class of $f\left(  R\right)  $ theories predicting acceleration. The role
of the coupling constant is briefly discussed.

\end{abstract}

\section{Introduction}

The standard inflationary idea requires that there be a period of slow-roll
evolution of a scalar field (the inflaton) during which, its potential energy
drives the universe in a quasi-exponential expansion. Besides a cosmological
constant, a nearly massless scalar field (quintessence) provides the simplest
mechanism to obtain accelerated expansion of the universe within General
Relativity. Therefore, scalar fields play a prominent role in the construction
of cosmological scenarios aiming to describe the evolution of the early and
the present universe.

Earlier investigations in scalar-field cosmology assumed a minimal coupling of
the scalar field, (see for example \cite{clw} and \cite{hcw} for models
containing both a perfect fluid of ordinary matter and a scalar field with an
exponential potential, the so-called \textquotedblleft
scaling\textquotedblright\ cosmologies; \cite{fay} and references therein for
scalar-tensor theories with exponential potential, \cite{hh} for a phase-space
analysis of the qualitative evolution of cosmological models with a scalar
field with positive or negative exponential potentials). Inclusion of non
minimal coupling increases the mathematical difficulty of the analysis, yet it
is important to consider non minimal coupling in scalar field cosmology
\cite{fuma}. As is stressed in \cite{fara1}, the introduction of non minimal
coupling is not a matter of taste; a large number of physical theories predict
the presence of a scalar field coupled to matter and we mention a few
important examples:

In the string effective action, the dilaton field is generally coupled to
matter in the Einstein frame \cite{gasp}. In scalar-tensor theories of gravity
\cite{fuma,lsf}, the action in the Einstein frame takes the form%
\begin{equation}
S=\int d^{4}x\sqrt{-g}\left\{  R-\left[  \left(  \partial\phi\right)
^{2}+2V\left(  \phi\right)  \right]  +2\chi^{-2}L_{\mathrm{m}}\left(
\widetilde{g}_{\mu\nu},\Psi\right)  \right\}  , \label{scte}%
\end{equation}
with%
\[
\widetilde{g}_{\mu\nu}=\chi^{-1}g_{\mu\nu},
\]
where $\chi=\chi\left(  \phi\right)  $ is the coupling function and matter
fields are collectively denoted by $\Psi$. In particular, for higher order
gravity (HOG) theories derived from Lagrangians of the form
\begin{equation}
f\left(  \widetilde{R}\right)  +2L_{\mathrm{m}}\left(  \widetilde{g}_{\mu\nu
},\Psi\right)  , \label{lagr}%
\end{equation}
it is well known that under the conformal transformation, $g_{\mu\nu
}=f^{\prime}\left(  \widetilde{R}\right)  \widetilde{g}_{\mu\nu}$, the field
equations reduce to the Einstein field equations with a scalar field $\phi$ as
an additional matter source. The conformal equivalence can be formally
obtained by conformally transforming the Lagrangian (\ref{lagr}) and the
resulting action becomes \cite{bbpst},
\[
S=\int d^{4}x\sqrt{-g}\left\{  R-\left[  \left(  \partial\phi\right)
^{2}+2V\left(  \phi\right)  \right]  +2e^{-2\sqrt{2/3}\phi}L_{\mathrm{m}%
}\left(  e^{-\sqrt{2/3}\phi}g_{\mu\nu},\Psi\right)  \right\}  .
\]
Therefore the Lagrangian of HOG theories is a particular case of the general
scalar-tensor Lagrangian with $\chi\left(  \phi\right)  =e^{\sqrt{2/3}\phi}$,
in equation (\ref{scte}). Non minimally coupling occurs also in models of
chameleon gravity \cite{khwe}, \cite{wate},
\[
S=\int d^{4}x\sqrt{-g}\left\{  R-\left[  \left(  \partial\phi\right)
^{2}+2V\left(  \phi\right)  \right]  +2L_{\mathrm{m}}\left(  \widetilde
{g}_{\mu\nu},\Psi\right)  \right\}  ,
\]
with%
\[
\widetilde{g}_{\mu\nu}=e^{2\beta\phi}g_{\mu\nu},
\]
where $\beta$ is a coupling constant. The same form of coupling has been
proposed in models of the so called coupled quintessence \cite{amen1} (see
also \cite{psc} for more general couplings and \cite{tmh} for a generalization
involving a scalar field coupled both to matter and a vector field).

Variation of the action (\ref{scte}) with respect to the metric $g$ yields the
field equations,
\begin{equation}
G_{\mu\nu}=T_{\mu\nu}\left(  g,\phi\right)  +T_{\mu\nu}^{\mathrm{m}}\left(
g,\Psi\right)  , \label{confm}%
\end{equation}
where $T_{\mu\nu}^{\mathrm{m}}$ is the matter energy momentum tensor. The
Bianchi identities imply that the total energy-momentum tensor is conserved
and therefore there is an energy exchange between the scalar field and
ordinary matter. In all the above examples, the conservation of their sum is
provided by the equations (compare to \cite{amen1}),
\[
\nabla^{\mu}T_{\mu\nu}^{\mathrm{m}}\left(  g,\Psi\right)  =QT^{\mathrm{m}%
}\nabla_{\nu}\phi,\ \ \ \ \nabla^{\mu}T_{\mu\nu}\left(  g,\phi\right)
=-QT^{\mathrm{m}}\nabla_{\nu}\phi,
\]
where $Q:=d\ln\chi/d\phi,$ depends in general on $\phi$ and $T^{\mathrm{m}}$
is the trace of the matter energy-momentum tensor, i.e., $T^{\mathrm{m}%
}=g^{\mu\nu}T_{\mu\nu}^{\mathrm{m}}\left(  g,\Psi\right)  $. Variation of $S$
with respect to $\phi$ yields the equation of motion of the scalar field,%
\begin{equation}
\square\phi-\frac{dV}{d\phi}=-QT^{\mathrm{m}}. \label{emsf}%
\end{equation}

In this paper we study the late time evolution of initially expanding flat FRW
models, with a scalar field coupled to matter and having a potential of the
form
\begin{equation}
V\left(  \phi\right)  =V_{1}e^{-\alpha\phi}+V_{2}e^{-\beta\phi},\label{pote}%
\end{equation}
where $\alpha,\beta$ are positive constants and $V_{1},V_{2}$ are constants of
arbitrary sign. Without loss of generality, we assume $0<\alpha<\beta.$ For
$0<\alpha=\beta$ the case reduces to a single exponential potential. We also
assume that the coupling coefficient is a constant, of order $Q\lesssim1$. The
double exponential potential is usually the asymptotic form of other
potentials. For example in Kaluza-Klein theories with $d$ extra dimensions
reformulated in the Einstein frame, $\alpha\ $and $\beta$ are $\sqrt
{2d/\left(  d+2\right)  }$ and $\sqrt{2\left(  d+2\right)  /d}$ respectively,
\cite{fuma}. The physical reason for the choice (\ref{pote}), is that in
quintessence models, the dark energy is the energy of a slowly varying scalar
field $\phi$ with equation of state $p_{\phi}=w\rho_{\phi},$ $w\simeq-1$. In
most of the models of dark energy, it is assumed that the cosmological
constant is zero and the potential energy, $V\left(  \phi\right)  ,$ of the
scalar field driving the present stage of acceleration, slowly decreases and
eventually vanishes as the field approaches the value $\phi=\infty$,
\cite{kali}. In this case, after a transient accelerating stage, the speed of
expansion of the universe decreases and the universe reaches Minkowski regime.
Double exponential potentials of the form (\ref{pote}) were investigated in
\cite{bcn,ss}. Solutions were obtained in \cite{glq} with the ansatz
$\dot{\phi}=\lambda H$; see also \cite{leon} for more general couplings. A
scalar field with a double exponential potential without coupling to matter
was investigated in \cite{lizs}. For exact solutions of a scalar field non
coupled to dust with single and double exponential potentials see \cite{rspc}.
Quintessence cosmologies of double exponential potentials in the absence of
matter were studied in \cite{jms} with the techniques of phase space analysis.
Coupled quintessence field with a double exponential potential and galileon
like correction was considered in \cite{aghs}.

The plan of the paper is as follows. In the next Section we write the field
equations for flat FRW models as a constrained four-dimensional dynamical
system. Assuming an initially expanding universe, we show that for potentials
(\ref{pote}) the scalar field almost always diverges to plus or minus infinity
as $t\rightarrow\infty,$ depending on the signs of $V_{1},V_{2}.$ Using
expansion-normalized variables, the system is written as a polynomial
three-dimensional system. In Section 3 we study the equilibrium points and
analyze the structure of the solutions. It is shown that under conditions on
the parameter space, the model is able to give an acceptable cosmological
history of our universe: a transient matter era followed by an accelerating
future attractor. In particular, if we assume that ordinary matter satisfies
plausible energy conditions, i.e., $\gamma\gtrsim1$, the scale factor during
the matter era evolves approximately as $a\sim t^{2/3}$, provided that the
coupling constant, $Q,$ takes very small values. In Section 4 we examine the
asymptotic form of a popular class of $f\left(  R\right)  $ theories
predicting acceleration; in the Einstein frame this theory is equivalent to a
scalar field with a double exponential potential and we discuss its
cosmological viability. Section 5 is a brief discussion on the acceptable
range of the coupling constant.

\section{Coupled scalar field model}

For homogeneous and isotropic flat spacetimes the field equations
(\ref{confm}) and (\ref{emsf}), (ordinary matter is described by a perfect
fluid with equation of state $p=(\gamma-1)\rho,$ where $0<\gamma<2$), reduce
to the Friedmann equation,
\begin{equation}
3H^{2}=\rho+\frac{1}{2}\dot{\phi}^{2}+V\left(  \phi\right)  ,\label{frie}%
\end{equation}
the Raychaudhuri equation,
\begin{equation}
\dot{H}=-\frac{1}{2}\dot{\phi}^{2}-\frac{\gamma}{2}\rho,\label{ray}%
\end{equation}
the equation of motion of the scalar field,
\begin{equation}
\ddot{\phi}+3H\dot{\phi}+V^{\prime}\left(  \phi\right)  =\frac{4-3\gamma}%
{2}Q\rho,\label{ems}%
\end{equation}
and the conservation equation,
\begin{equation}
\dot{\rho}+3\gamma\rho H=-\frac{4-3\gamma}{2}Q\rho\dot{\phi}.\label{conss}%
\end{equation}
We adopt the metric and curvature conventions of \cite{wael}. $a\left(
t\right)  $ is the scale factor, an overdot denotes differentiation with
respect to time $t,$ $H=\dot{a}/a$ and units have been chosen so that
$c=1=8\pi G.$ Here $V\left(  \phi\right)  $ is the potential energy of the
scalar field and $V^{\prime}\left(  \phi\right)  =dV/d\phi.$ Interaction terms
between the two matter components of the form $-\alpha\rho\dot{\phi}$ as in
(\ref{conss}) with a simple exponential potential, were firstly considered in
\cite{bico} (see also \cite{bclm}). Although there is an energy exchange
between the fluid and the scalar field, it is easy to see that the set,
$\rho>0,$ is invariant under the flow of (\ref{ray})-(\ref{conss}), therefore
$\rho$ is nonzero if initially $\rho\left(  t_{0}\right)  $ is nonzero; this
trivial physical demand is not satisfied if one assumes arbitrary interaction
terms, cf. \cite{miri3}. 

\begin{figure}[th]
\begin{center}
\includegraphics{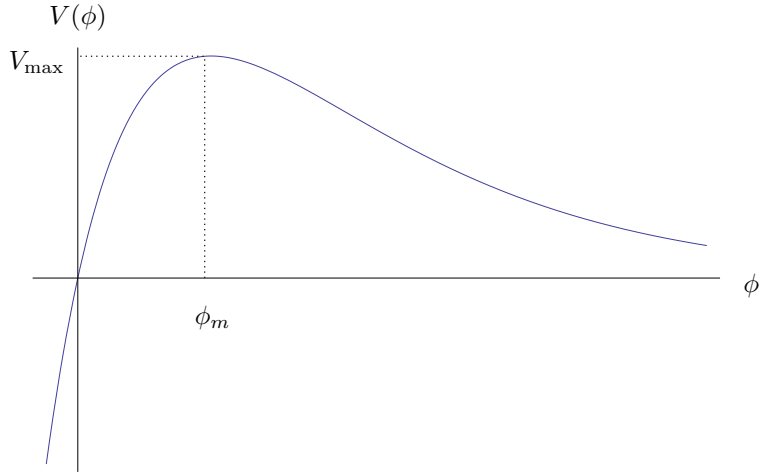}
\end{center}
\caption{Potentials (\ref{pote}) with $V_{1}>0,V_{2}<0$ have a local maximum
at some $\phi_{m}$ and diverge to minus infinity as $\phi\rightarrow-\infty$.
In this figure, $V_{2}=-V_{1}<0$. }%
\label{fig1}%
\end{figure}

As is explained in the last paragraph of the Appendix, the physically
interesting cases are $V_{1},V_{2}>0$ or $V_{1}>0,~V_{2}<0$. The dynamical
system (\ref{ray})-(\ref{conss}) has for $V_{1}>0,~V_{2}<0,$ only one finite
equilibrium point, $\left(  \phi=\phi_{m},\dot{\phi}=0,\rho=0,H=\sqrt{V_{\max
}/3}\right)  ,$ see Figure 1. It represents de Sitter solutions and is easy to
see that it is unstable. It is known that for potentials having a maximum, the
field near the top of the potential corresponds to the tachyonic (unstable)
mode with negative mass squared \cite{linde,kali,kali1}. The other asymptotic
states of the system correspond to the points at infinity, $\phi\rightarrow
\pm\infty.$

For potentials (\ref{pote}) with $V_{1},V_{2}>0,$ it can be shown the global
result that, for expanding flat models, $\phi\rightarrow\infty$ as
$t\rightarrow\infty.$ In fact, the following slightly stronger result holds,
which generalizes Proposition 4 in \cite{miri2}.

\begin{proposition}
\label{prop1}Let $V$ be a potential function with the following properties: 1.
$V$ is non-negative. 2. $V^{\prime}$ is continuous and $V^{\prime}\left(
\phi\right)  <0$. 3. If $A\subseteq\mathbb{R}$ is such that $V$ is bounded on
$A$, then $V^{\prime}$ is bounded on $A$. Then $\lim_{t\rightarrow+\infty}%
\dot{\phi}=0=\lim_{t\rightarrow+\infty}\rho,$ and $\lim_{t\rightarrow+\infty
}\phi=+\infty$.
\end{proposition}

\begin{proof}
Since $V\left(  \phi\right)  \geq0,$ it follows from (\ref{frie}) that $H$ is
never zero, thus it cannot change sign. Hence, $H$ is always non-negative if
$H\left(  t_{0}\right)  >0$. Furthermore, $H$ is decreasing in view of
(\ref{ray}), thus$\ H\left(  t\right)  \leq H\left(  t_{0}\right)  $, for all
$t\geq t_{0}$. We then deduce from (\ref{frie}) that each of the terms
$\rho,\frac{1}{2}\dot{\phi}^{2}$ and $V$ is bounded by $3H\left(
t_{0}\right)  ^{2}$. Since $H$ is decreasing, $\exists\lim_{t\rightarrow
+\infty}H=\eta\geq0$, therefore (\ref{ray}) implies that%
\begin{equation}
\frac{1}{2}\int_{t_{0}}^{+\infty}\left(  \dot{\phi}^{2}+\gamma\rho\right)
dt=H\left(  t_{0}\right)  -\eta<+\infty. \label{*}%
\end{equation}
In general, if $f$ is a non-negative function, the convergence of
$\int\nolimits_{t_{0}}^{\infty}f\left(  t\right)  dt$ does not imply that
$\lim_{t\rightarrow\infty}f\left(  t\right)  =0,$ unless the derivative of $f$
is bounded. In our case and setting $\lambda=\left(  4-3\gamma\right)  Q$,
\begin{align*}
\frac{d}{dt}\left(  \dot{\phi}^{2}+\gamma\rho\right)   &  =-6H\dot{\phi}%
^{2}-2\dot{\phi}V^{\prime}\left(  \phi\right)  -3\gamma^{2}\rho H+\lambda
\left(  1-\frac{\gamma}{2}\right)  \rho\dot{\phi}\\
&  \leq-2\dot{\phi}V^{\prime}\left(  \phi\right)  +\lambda\left(
1-\frac{\gamma}{2}\right)  \rho\dot{\phi}.
\end{align*}
As we already remarked, $\dot{\phi}$ and $\rho$ are bounded; also, by our
assumption on $V$, $V^{\prime}\left(  \phi\right)  $ is bounded. We conclude
that the derivative of the function $\dot{\phi}^{2}+\gamma\rho$ is bounded
from above and therefore, (\ref{*}) implies that $\lim_{t\rightarrow\infty
}\dot{\phi}\left(  t\right)  ^{2}=0\;$and$\;\lim_{t\rightarrow\infty}%
\rho\left(  t\right)  =0.$

The proof that, $\lim_{t\rightarrow+\infty}\phi=+\infty,$ follows after
suitable adaptation of the arguments used in Proposition 4 in \cite{miri2}.
\end{proof}

If in addition, $\lim_{\phi\rightarrow+\infty}V\left(  \phi\right)  =0,$ as is
the case of the double exponential potential (\ref{pote}), then we conclude
that $H\rightarrow0$ as $t\rightarrow\infty.$

The case $V_{1}>0,V_{2}<0,$ is more delicate and the asymptotic state depends
on the initial conditions. (i) If initially $\phi_{0}>\phi_{m},$ and
$3H\left(  t_{0}\right)  ^{2}<V_{\max},$ then from (\ref{frie}), $V\left(
\phi\right)  $ remains less than $V_{\max}$ since $H$ is decreasing. We
conclude that $V\left(  \phi\left(  t\right)  \right)  <V_{\max}$ for all
$t\geq t_{0}$, thus $\phi$ cannot pass to the left of $\phi_{m}$. In the
interval $\left(  \phi_{m},+\infty\right)  $ the potential satisfies the
assumptions of the above Proposition and therefore, $\phi\rightarrow\infty$ as
$t\rightarrow\infty.$ (ii) If initially $\phi_{0}<\phi_{m},$ and $\dot{\phi
}_{0}$ is larger than the critical value $\dot{\phi}_{\mathrm{crit}}>0$, which
allows for $\phi$ to pass on the right of $\phi_{m},$ then the conclusions of
case (i) hold. (iii) Finally, suppose that initially $\phi_{0}<\phi_{m},$ and
$\dot{\phi}_{0}$ is less than the critical value $\dot{\phi}_{\mathrm{crit}%
}>0,$ i.e., $-\infty<\dot{\phi}_{0}<\dot{\phi}_{\mathrm{crit}}.$ From
(\ref{ray}), $H$ is monotonically decreasing and not bounded below from zero,
hence eventually $H$ may change sign. We cannot use the same argument as in
Proposition \ref{prop1} concerning the asymptotic behavior of $\dot{\phi
}\left(  t\right)  ^{2}$ and$\;\rho\left(  t\right)  ,$ since $V$ and
$V^{\prime}$ are not bounded. Suppose, firstly, that $\lim_{t\rightarrow
+\infty}H=\eta,$ where $\eta\ $is finite. But, an asymptotic state of the
form, $\mathbf{p}=\left(  H=\eta,\rho=\rho_{\ast},\dot{\phi}=\dot{\phi}_{\ast
},\phi=\phi_{\ast}\right)  ,$ is impossible, i.e., the point $\mathbf{p}$
cannot be an equilibrium point of the dynamical system (\ref{ray}%
)-(\ref{conss}) for $\phi_{\ast}<\phi_{m}.$ Although we cannot exclude
periodic orbits, or strange attractors as $\omega-$limit sets for our system,
numerical experiments suggest that, $H$ diverges to $-\infty$. If this is the
case, it can be shown that $H$ diverges to $-\infty,$ in a finite time.
Suppose on the contrary that, $\lim_{t\rightarrow+\infty}H=-\infty.\ $Since
$\gamma<2$,%
\[
3H^{2}=\frac{\dot{\phi}^{2}}{2}+\rho+V(\phi)<\frac{\dot{\phi}^{2}+\gamma\rho
}{\gamma}+V(\phi)=-\frac{2\dot{H}}{\gamma}+V(\phi),
\]
hence,%
\begin{equation}
3<-\frac{2\dot{H}}{\gamma H^{2}}+\frac{V(\phi)}{H^{2}}. \label{ineq}%
\end{equation}
Taking limits as $t\rightarrow+\infty$, and since $V(\phi)$ is bounded from
above, $\lim_{t\rightarrow+\infty}V(\phi)/H^{2}\leq0$. Inequality (\ref{ineq})
implies that $\lim_{t\rightarrow+\infty}\left(  -\dot{H}/H^{2}\right)
\geq3\gamma/2$, which is impossible, since $-\dot{H}/H^{2}=d/dt\left(
1/H\right)  $ and $1/H\rightarrow0$. In view of (\ref{frie}), $\dot{\phi}%
^{2}+\gamma\rho\ $also diverges to infinity. Again, an asymptotic state of the
form, $H=-\infty,\dot{\phi}^{2}+\gamma\rho=\infty\ $and $\phi=$ finite is
impossible, therefore $\phi\ $diverges to $-\infty$ in a finite time. The
above arguments, supported by numerical investigation, establish the following
result, although we were unable to prove it rigorously:

\begin{proposition}
Let $V$ be a $C^{1}$ potential function with the following properties: 1. $V$
is negative and monotonically increasing for $\phi<0,$ with $\lim
_{\phi\rightarrow-\infty}V\left(  \phi\right)  =-\infty$. 2. $V$ has a global
maximum at some $\phi_{m}>0$ . Suppose that the following initial conditions
hold: $H\left(  t_{0}\right)  >0,$ $\phi\left(  t_{0}\right)  <\phi_{m},$ and
$-\infty<\dot{\phi}\left(  t_{0}\right)  <\dot{\phi}_{\mathrm{crit}},$ where
$\dot{\phi}_{\mathrm{crit}}>0,$ is the critical value which allows for $\phi$
to pass to the right of $\phi_{m}$. Then $H\ $and $\phi\ $diverge
to\ $-\infty\ $in a finite time.
\end{proposition}

This result generalizes previous investigations indicating that negative
potentials may drive a flat initially expanding universe to recollapse, see
\cite{hh,ffkl,gma}. Negative potentials appear also in ekpyrotic models (see
for example \cite{bko} and references therein and \cite{kowa} with multiple fields).

The function (\ref{pote}) belongs to the class of multi-exponential potentials
of the form
\[
V\left(  \phi\right)  =\sum_{i=1}^{N}V_{i}e^{-k_{i}\phi},
\]
which arise as a special case of generalized models with multiple fields
studied in the context of assisted inflation (see for example \cite{coho}; for
an elegant mathematical generalization see \cite{cnr}). There exists a well
established mathematical procedure for the investigation of scalar field
cosmologies with exponential potentials in the context of dynamical systems
theory \cite{clw,wael}. It consists in the introduction of the so called,
expansion normalized variables by defining
\begin{equation}
x=\frac{\dot{\phi}}{\sqrt{6}H},~~y=\sqrt{\frac{V_{1}e^{-\alpha\phi}}{3H^{2}}%
},~~z=\sqrt{\frac{V_{2}e^{-\beta\phi}}{3H^{2}}},~~\Omega=\frac{\rho}{3H^{2}},
\label{env}%
\end{equation}
and a new time variable $\tau=\ln a.$ The Friedmann equation (\ref{frie})
imposes the constraint
\begin{equation}
\Omega=1-\left(  x^{2}+y^{2}+z^{2}\right)  , \label{cons1}%
\end{equation}
to the state vector $\left(  x,y,z,\Omega\right)  $. This equation can be used
to eliminate $\Omega$ from the evolution equations and we end up with a
three-dimensional dynamical system,%

\begin{align}
x^{\prime}  &  =\sqrt{6}Q-\frac{3}{2}\sqrt{\frac{3}{2}}\gamma Q+\left(
\frac{3\gamma}{2}-3\right)  x+\left(  \frac{3}{2}\sqrt{\frac{3}{2}}%
\gamma-\sqrt{6}\right)  Qx^{2}+\nonumber\\
&  +\left(  3-\frac{3\gamma}{2}\right)  x^{3}+\left(  \sqrt{\frac{3}{2}}%
\alpha-\sqrt{6}Q+\frac{3}{2}\sqrt{\frac{3}{2}}\gamma Q\right)  y^{2}%
+\nonumber\\
&  +\left(  \sqrt{\frac{3}{2}}\beta-\sqrt{6}Q+\frac{3}{2}\sqrt{\frac{3}{2}%
}\gamma Q\right)  z^{2}-\frac{3}{2}\gamma xy^{2}-\frac{3}{2}\gamma
xz^{2},\nonumber\\
y^{\prime}  &  =y\left(  \frac{3\gamma}{2}-\sqrt{\frac{3}{2}}\alpha x+\left(
3-\frac{3\gamma}{2}\right)  x^{2}-\frac{3\gamma}{2}y^{2}-\frac{3\gamma}%
{2}z^{2}\right)  ,\label{sys3}\\
z^{\prime}  &  =z\left(  \frac{3\gamma}{2}-\sqrt{\frac{3}{2}}\beta x+\left(
3-\frac{3\gamma}{2}\right)  x^{2}-\frac{3\gamma}{2}y^{2}-\frac{3\gamma}%
{2}z^{2}\right)  ,\nonumber
\end{align}
where
\begin{equation}
x^{2}+y^{2}+z^{2}\leq1, \label{const}%
\end{equation}
and a prime denotes derivative with respect to $\tau$. Note that $y$ and $z$
can take both real and pure imaginary values, depending on the signs of
$V_{i}$. With this choice we avoid to have four different dynamical systems
(see however \cite{hh} where real normalized variables are used). For
$V_{1},V_{2}>0$, the phase space (\ref{const}) is the closed unit ball in
$\mathbb{R}^{3}$. For $V_{1}>0$ and $V_{2}<0$, the phase space is the one
sheet hyperboloid $x^{2}+y^{2}-(\operatorname{Im}z)^{2}=1$ and its interior.
The resulting dynamical system depends on four parameters $(\gamma
,\alpha,\beta,Q)$. Using (\ref{ray}), the effective equation of state,
\[
w_{\mathrm{eff}}=-1-\frac{2\dot{H}}{3H^{2}},
\]
is written in terms of the new variables as,
\[
w_{\mathrm{eff}}=-1+2x^{2}+\gamma\Omega.
\]

\section{Cosmologically acceptable solutions}

By inspection, system (\ref{sys3}) is symmetric under reflection, with respect
to the planes $x-z$ and $x-y$. The planes $y=0$ and $z=0$ are invariant sets
for the system (\ref{sys3}). The full list and analysis of the critical points
of our system is presented in the Appendix. In this section, we discuss only
these equilibria which allow for a viable cosmological history of the
universe. In Table 1 are shown the equilibria for $V_{1}>0$ and
\[
\alpha<\sqrt{2},~\gamma\leq1,~\left(  4-3\gamma\right)  Q\in\left(
\max\left\{  0,2\left(  \alpha^{2}-3\gamma\right)  /\alpha\right\}  ,\sqrt
{6}\left(  2-\gamma\right)  \right)  .
\]

\begin{center}
\begin{table}[ptb]
\caption{Equilibrium Points}%
\begin{tabular}
[c]{lllll}\hline\hline
Label & $(x,y,z)$ & $\Omega$ & Stability & $a(t)$\\\hline
$\mathcal{A}_{\pm}$ & $\left(  \pm1,0,0\right)  $ & $0$ & Unstable & $t^{1/3}%
$\\
$\mathcal{B}$ & $\left(  \frac{\left(  4-3\gamma\right)  Q}{\sqrt{6}\left(
2-\gamma\right)  },0,0\right)  $ & $1-\frac{\left(  4-3\gamma\right)  ^{2}
Q^{2}}{6 \left(  2-\gamma\right)  ^{2} }$ & Saddle & $t^{4 \left(  2-
\gamma\right)  / \left(  6 \gamma\left(  2- \gamma\right)  + \left(  4-3
\gamma\right)  ^{2} Q^{2} \right)  }$\\
$\mathcal{C}_{\pm}$ & $\left(  \frac{\alpha}{\sqrt{6}},\pm\sqrt{1-\frac
{\alpha^{2}}{6}},0\right)  $ & $0$ & Stable & $t^{2/ \alpha^{2}}$\\
$\mathcal{D}_{\pm}$ & $\left(  0,\pm\sqrt{\frac{\beta}{\beta- \alpha}},
\pm\sqrt{\frac{\alpha}{\alpha- \beta}}\right)  $ & $0$ & Saddle & $e^{t}%
$\\\hline\hline
\end{tabular}
\end{table}
\end{center}

The two critical points $\mathcal{A}_{\pm}$ correspond to kinetic dominated
solutions which are unstable and are only expected to be relevant at early
times. Point $\mathcal{B}$ represents a type of scaling solution, i.e., the
kinetic energy density of the scalar field remains proportional to that of the
perfect fluid. Points $\mathcal{C}_{\pm}$ are accelerated only for $V_{1}>0$.
They correspond to scalar field dominated solutions which exist for
sufficiently flat potentials, $\alpha<\sqrt{6}.$ These are the same
conclusions as in \cite{gimi} for an exponential potential and $Q=\sqrt{2/3},$
and also in \cite{clw}, \cite{hh} and \cite{lizs} and in the case of a scalar
field non coupled to matter, although the ranges of the parameters
$(\alpha,\gamma)$ are different. Points $\mathcal{D}_{\pm}$ exist only in
models with $V_{1}>0,~V_{2}<0$. They correspond to the unstable state $\left(
\phi=\phi_{m},\dot{\phi}=0,\rho=0,H=\sqrt{V_{\max}/3}\right)  $ and represent
de Sitter solutions.

A successful cosmological model should comprise an accelerating solution as a
future attractor. It is evident that points $\mathcal{C}_{\pm}$, could satisfy
the condition for acceleration, $w_{\mathrm{eff}}<-1/3,$ provided that
$\alpha<\sqrt{2},$ (compare with the conclusions in \cite{clw}). From now on
we assume this range for the parameter $\alpha$. Moreover, the equilibria
$\mathcal{C}_{\pm}$, are stable for all physically interesting values of
$\gamma$. For a cosmological theory to be acceptable, it has to possess a
matter dominated epoch followed by a late time accelerated attractor. The
saddle character of point $\mathcal{B}$, implies that it represents a
transient phase and therefore, it is a good candidate for a matter point,
provided that $\Omega$ is close to one. This happens only for very small
values of the coupling parameter $Q$ and for $\gamma$ close to one. Another
way to see this, is the following. During the matter era, the scale factor has
to expand approximately as $a\sim t^{2/3}$. The scale factor near
$\mathcal{B}$ evolves as $a\sim t^{\frac{2}{3(w_{\mathrm{eff}}+1)}},$
therefore, $w_{\mathrm{eff}}$, has to be close to zero. As seen in Table 1,
$a\left(  t\right)  $ at $\mathcal{B}$, evolves as $t^{2/3}$ when $Q$ takes
the values%
\begin{equation}
Q=\frac{\sqrt{6\left(  2-\gamma\right)  \left(  1-\gamma\right)  }}{\left(
4-3\gamma\right)  },~\gamma\leq1. \label{Q}%
\end{equation}
Therefore, the realistic value $\gamma=1,$ corresponding to dust, is
incompatible to scalar field coupled to matter, i.e., the coupling parameter
$Q$ must be zero (see also \cite{aqtw}). On the other hand, (\ref{ems}) and
(\ref{conss}) imply that for $\gamma=4/3,$ the value of $Q$ is undetermined.
Below we summarize our results for the particular values $\gamma=1,4/3,2/3.$

A. Dust ($\gamma=1$). The critical points of our system are those of Table 1
for $\alpha<\sqrt{2},~\beta>\alpha,~Q=0$. Note that the future attractors
$\mathcal{C}_{\pm}$ have non phantom acceleration for every value of $\alpha$
in the interval $(0,\sqrt{2})$. A cosmologically acceptable trajectory should
pass near $\mathcal{B}$ and finally land on one of the points $\mathcal{C}%
_{\pm}$, depending on the initial conditions. Note that $\mathcal{A}_{\pm},$
$\mathcal{B}$ and $\mathcal{C}_{\pm}$ lie on the invariant plane $z=0$ and
$\mathcal{C}_{\pm}$ exist only in potentials with $V_{1}>0$. We consider the
projection of the system (\ref{sys3}) on that plane. The phase portrait is
shown in Figure 2 and is the same in both cases where the phase space is a
sphere ($V_{2}>0$), or a one sheet hyperboloid ($V_{2}<0$).

\begin{figure}[th]
\begin{center}
\includegraphics[scale=0.9]{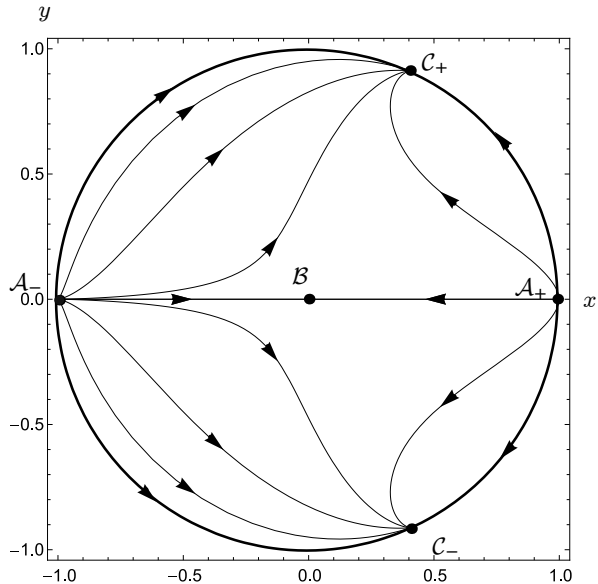}
\end{center}
\caption{Phase portrait of the projected three-dimensional system on the
invariant set $z=0$.}%
\label{fig2}%
\end{figure}

B. Radiation ($\gamma=4/3$). The case of $\gamma=4/3$ corresponds to
radiation, and therefore there is no matter point with a scale factor $a\sim
t^{2/3}$. Instead, point $\mathcal{B}$, which coincides with the origin
$(0,0,0)$, now represents the well-known radiation dominated solution, $a\sim
t^{1/2}$, as a transient phase. $\mathcal{C}_{\pm}$ are future attractors for
$\alpha< \sqrt{2}$.

C. The value $\gamma=2/3$ corresponds to ordinary matter marginally satisfying
the strong energy condition. Eq. (\ref{Q}) implies $Q=\sqrt{2/3}$. An
acceptable trajectory exists for $\alpha<\sqrt{2}$. For these values of
$\alpha$ and $Q$, points $\mathcal{A}_{\pm}$ are always unstable. Point
$\mathcal{B}\equiv\left(  1/2,0,0\right)  $, corresponds to the transient
matter era, with $\Omega=3/4$. The accelerated points $\mathcal{C}_{\pm}$ are
future attractors.

Throughout this paper we do not consider the case, $\alpha\beta<0,$ for the
potentials (\ref{pote}). The reason is that for $\alpha\beta<0,$ and
$V_{1},V_{2}>0,$ the function $V\left(  \phi\right)  $ in (\ref{pote}) has a
strictly positive minimum, say $V_{\min},$ and the de Sitter solution with
$H=\sqrt{V_{\min}/3}$, is the future attractor for the system, \cite{gimi}.
This follows directly either from the original equations (\ref{ray}%
)-(\ref{conss}), or from the system (\ref{sys3}) written in the new variables.
Moreover, it is easy to see that a matter era represented by a saddle
equilibrium $\mathcal{B}$, precedes the final accelerated epoch.

\section{Asymptotic form of some $f\left(  R\right)  $ theories predicting
acceleration}

A large class of dynamical dark energy models is based on the large-distance
modification of gravity (see \cite{fets} for recent reviews). For example, in
the context of $f\left(  R\right)  $ gravity theories the models
$f(R)=R-\mu^{2(n+1)}/R^{n}$, where $\mu>0,n>1$, were proposed to explain the
late-time cosmic acceleration \cite{ccct,cdtt}. The obvious idea is the
introduction of modifications to the Einstein-Hilbert Lagrangian which become
important at low curvatures. For these models the potential functions in the
Einstein frame have the form,%

\begin{equation}
V_{n}(\phi)=\frac{\mu^{2}(n+1)n^{1/(n+1)}(e^{\sqrt{2/3}\phi}-1)^{n/(n+1)}%
}{2ne^{2\sqrt{2/3}\phi}}. \label{pot}%
\end{equation}
These functions are defined only for $\phi\geq0,$ and their behavior is
similar to that indicated in Figure 1, i.e., they have a local maximum at some
$\phi_{m}$ depending on $n,$ and for large $\phi$ they approach zero
exponentially. As $n\rightarrow\infty$ the potentials (\ref{pot}) approach the function,%

\begin{equation}
V(\phi)=\frac{\mu^{2}}{2}\left(  e^{-\sqrt{2/3}\phi}-e^{-2\sqrt{2/3}\phi
}\right)  , \label{pot2}%
\end{equation}
corresponding to the asymptotic form of these theories, \cite{cdtt}. Thus,
(\ref{pot2}) is a particular case of (\ref{pote}) with $\beta=2\alpha
=2\sqrt{2/3}$, $V_{1}=-V_{2}=\mu^{2}/2>0$, cf. Figure 1. Note that for large
$\phi,$ $V$ in (\ref{pot2}) behaves similarly to $V_{n}$ in (\ref{pot}). In
contrast to the family (\ref{pot}), $V$ in (\ref{pot2}) has the nice property
that it is defined for all $\phi\in\mathbb{R}$. As mentioned in the
introduction, the coupling coefficient takes the value $Q=\sqrt{2/3}$,
regardless of the form of $f(R),$ \cite{apt}.

The constraint (\ref{const}) implies that the phase space is the set
$x^{2}+y^{2}-(\operatorname{Im}z)^{2}\leq1$. There are up to seven critical
points for that system, depending on the value of $\gamma$.

\begin{center}%
\begin{tabular}
[t]{llllll}\hline\hline
Label & $(x,y,z)$ & $\Omega$ & Existence & Stability & $a(t)$\\\hline
$\mathcal{A}_{\pm} $ & $(\pm1,0,0)$ & $0$ & always & unstable & $t^{1/3}$\\
$\mathcal{B} $ & $\left(  \frac{4-3 \gamma}{3 \left(  2- \gamma\right)  },0,0
\right)  $ & $\frac{4 \left(  5- 3\gamma\right)  }{9\left(  2- \gamma\right)
^{2}}$ & $\gamma\leq5/3$ & saddle & $t^{3 \left(  2- \gamma\right)  / \left(
8 -3 \gamma\right)  }$\\
$\mathcal{C}_{\pm} $ & $\left(  \frac{1}{3}, \pm\frac{2\sqrt{2}}{3},0 \right)
$ & $0$ & always & stable & $t^{3}$\\
$\mathcal{D}_{\pm}$ & $(0, \pm\sqrt{2}, \pm i) $ & $0$ & always & saddle &
$e^{t}$\\\hline\hline
\end{tabular}

\end{center}

Points $\mathcal{C}_{\pm}$ are future attractors and have non phantom
acceleration with $w_{\mathrm{eff}}=-7/9$. However, in the case of dust,
$\gamma=1$, the scale factor at matter point $\mathcal{B}$ evolves as $a\sim
t^{3/5},$ rather than the usual $a\sim t^{2/3}$. The scale factor evolves
\textquotedblleft correctly\textquotedblright\ only for $\gamma=2/3$. The
absence of the standard matter epoch is associated with the fact that matter
is strongly coupled to gravity. This result is in agreement with the general
conclusions in \cite{apt}, \cite{agpt}, \cite{amen2}, that these $f\left(
R\right)  $ dark energy models are not cosmologically viable.

\section{Conclusion}

In this paper we have focused on a general treatment of a scalar field with a
double exponential potential non minimally coupled to a perfect fluid. A full
analysis of the equilibrium points of the resulted dynamical system is quite
complicated, yet it revealed that the model predicts a late accelerated phase
of the universe for a wide range of the parameters, $\alpha,\beta,\gamma$ and
$Q$. Moreover, there exists transient solutions representing a matter era,
preceding the accelerating attractor. However, in most cases the scale factor
near these transient phases evolves as $a\left(  t\right)  \sim t^{q\left(
Q\right)  },$ where the exponent $q$ is in general different from the usual
$2/3.$ The \textquotedblleft wrong\textquotedblright\ matter epoch is
associated with the fact that for values of $Q$ of order unity, matter is
strongly coupled to gravity. A coupling constant of order unity means that
matter feels an additional scalar force as strong as gravity itself, cf.
\cite{apt}. Assuming that ordinary matter satisfies plausible energy
conditions, i.e., $\gamma\gtrsim1,$ the coupling constant, $Q,$ has to be very
small; more precisely, $q\left(  Q\right)  \rightarrow2/3$, only for
$Q\rightarrow0.$ Therefore, only a very weak coupling of the scalar field to
ordinary matter can lead to acceptable cosmological histories of the universe.
This surprising result, indicates that cosmological evolution imposes strict
constraints on the choice of the correct Lagrangian of a gravity theory. In
this study we restricted ourselves to constant couplings; had we let $Q$ to be
a function of $\phi$, the dimension of the dynamical system would have
increased by one. In that case, it would be very interesting to see if the
dynamics leads to a very tiny value of $Q$ at late times. Such a result could
lead to a generalization of the attractor mechanism of scalar-tensor theories
towards general relativity, found by Damour and Nordtvedt in the case of a
massless scalar field \cite{dano}.

\section*{Acknowledgements}

We thank N. Hadjisavvas and S. Cotsakis for useful comments.

\renewcommand{\theequation}{A.\arabic
{equation}}

\setcounter{equation}{0}

\section*{Appendix}

We present here the full analysis of the stability of the system (\ref{sys3}).
The critical points are listed in Table 2.

\begin{center}
\begin{table}[ptb]
\caption{Critical Points}%
\begin{tabular}
[c]{llll}\hline\hline
Label & $(x,y,z)$ & $\Omega$ & $w_{\text{eff}}$\\\hline
$\mathcal{A}_{\pm}$ & $\left(  \pm1,0,0 \right)  $ & $0$ & $1$\\
$\mathcal{B}$ & $\left(  \frac{\left(  4-3\gamma\right)  Q}{\sqrt{6}\left(  2-
\gamma\right)  },0,0\right)  $ & $1-\frac{\left(  4-3\gamma\right)  ^{2}Q^{2}%
}{6\left(  2- \gamma\right)  ^{2}} $ & $-1 + \gamma+ \frac{\left(  4-3
\gamma\right)  ^{2} Q^{2}}{6 \left(  2- \gamma\right)  }$\\
$\mathcal{C}_{\pm}$ & $\left(  \frac{\alpha}{\sqrt{6}},\pm\sqrt{1-
\frac{\alpha^{2}}{6}},0 \right)  $ & $0$ & $-1+\frac{\alpha^{2}}{3}$\\
$\mathcal{D}_{\pm}$ & $\left(  0,\pm\sqrt{\frac{\beta}{ \beta- \alpha}}
,\pm\sqrt{\frac{\alpha}{\alpha- \beta}} \right)  $ & $0$ & $-1$\\
$\mathcal{D^{\prime}}_{\pm}$ & $\left(  0,\pm\sqrt{\frac{\beta}{ \beta-\alpha
}},\mp\sqrt{\frac{\alpha}{\alpha- \beta}} \right)  $ & $0$ & $-1$\\
$\mathcal{E}_{\pm}$ & $\left(  u_{\alpha},\pm v_{\alpha},0\right)  $ &
$\omega_{\alpha}$ & $-1+ \sqrt{ \frac{2}{3}} \alpha u_{\alpha} $\\
$\mathcal{F}_{\pm}$ & $\left(  \frac{\beta}{\sqrt{6}},0,\pm\sqrt{1-
\frac{\beta^{2}}{6}} \right)  $ & $0$ & $-1+ \frac{ \beta^{2}}{3}$\\
$\mathcal{G}_{\pm}$ & $\left(  u_{\beta},0,\pm v_{\beta} \right)  $ &
$\omega_{\beta}$ & $-1+ \sqrt{ \frac{2}{3}} \beta u_{\beta} $\\\hline\hline
where & $u_{\alpha} = \frac{\sqrt{6}\gamma}{2\alpha-( 4- 3\gamma) Q},$ &  & \\
& $v_{\alpha} =\sqrt{\frac{\left(  4-3 \gamma\right)  ^{2} Q^{2} -2
\alpha\left(  4-3 \gamma\right)  Q +6 \gamma\left(  2- \gamma\right)
}{\left(  2\alpha- \left(  4-3\gamma\right)  Q\right)  ^{2}}},$ &  & \\
& $\omega_{\alpha} =\frac{2\left(  2\alpha^{2}-6\gamma-\alpha\left(
4-3\gamma\right)  Q\right)  }{\left(  2\alpha-\left(  4- 3\gamma\right)
Q\right)  ^{2}},$ &  & \\
& and similarly for $u_{\beta}, v_{\beta}, \omega_{\beta}$. &  & \\
&  &  &
\end{tabular}
\end{table}
\end{center}

We assume that $0<\alpha<\beta$. The case $0<\beta<\alpha$, is a mere renaming
of some of the equilibrium points. According to the definition (\ref{env}),
the modulus of $z$ lies between $0$ and the absolute value of $y$. Therefore,
points $\mathcal{D^{\prime}}_{\pm},~\mathcal{F}_{\pm}$ and $\mathcal{G}_{\pm}$
are not acceptable. The eigenvalues of the remaining equilibria are presented
in the Table 3.

\begin{center}
\begin{table}[ptb]
\caption{Eigenvalues}%
\begin{tabular}
[c]{ll}\hline\hline
Label & Eigenvalues\\\hline
$\mathcal{A}_{+}$ & $3-\sqrt{\frac{3}{2}}\alpha, ~3-\sqrt{\frac{3}{2}}\beta,~
6-3\gamma-\sqrt{\frac{3}{2}} \left(  4-3 \gamma\right)  Q$,\\
$\mathcal{A}_{-}$ & $3+\sqrt{\frac{3}{2}}\alpha, ~3+\sqrt{\frac{3}{2}}\beta,~
6-3\gamma+\sqrt{\frac{3}{2}} \left(  4-3 \gamma\right)  Q$,\\
$\mathcal{B}$ & $\frac{ \left(  4-3 \gamma\right)  ^{2} Q^{2} -2 \alpha\left(
4-3 \gamma\right)  Q +6 \gamma\left(  2 -\gamma\right)  }{4 \left(  2-
\gamma\right)  }, ~\frac{ \left(  4-3 \gamma\right)  ^{2} Q^{2} -2
\beta\left(  4-3 \gamma\right)  Q +6 \gamma\left(  2 -\gamma\right)  }{4
\left(  2- \gamma\right)  }, ~\frac{ \left(  4-3 \gamma\right)  ^{2} Q^{2} -6
\left(  2 -\gamma\right)  ^{2} }{4 \left(  2- \gamma\right)  }$\\
$\mathcal{C}_{\pm}$ & $\frac{\alpha^{2} -6}{2},~ \frac{\alpha\left(
\alpha-\beta\right)  }{2},~ \frac{2\alpha^{2} -6 \gamma- \alpha\left(  4-3
\gamma\right)  Q}{2}$\\
$\mathcal{D}_{\pm}$ & $\frac{1}{2} \left(  -3 \pm\sqrt{9+12 \alpha\beta}
\right)  ,~ -3 \gamma$\\
$\mathcal{E}_{\pm}$ & $\frac{3 \left(  \alpha-\beta\right)  \gamma}{2\alpha-
\left(  4-3 \gamma\right)  Q},~\frac{ \sigma\pm\sqrt{ \sigma^{2} -4 \delta}}{2
\left(  2\alpha- \left(  4-3 \gamma\right)  Q \right)  ^{2}} $\\\hline\hline
where & $\sigma= 3\left(  2\alpha-\left(  4-3\gamma\right)  Q\right)  \left(
\left(  4-3\gamma\right)  Q-\alpha\left(  2-\gamma\right)  \right)  ,$\\
and & $\delta= \frac{3}{2}\left(  2\alpha-\left(  4-3\gamma\right)  Q\right)
^{2}\left(  2\alpha^{2}-6\gamma-\alpha\left(  4-3\gamma\right)  Q\right)  $\\
& $\left(  \left(  4-3\gamma\right)  ^{2}Q^{2}-2\alpha\left(  4-3\gamma
\right)  Q+6\gamma\left(  2-\gamma\right)  \right)  .$%
\end{tabular}
\end{table}
\end{center}

As mentioned in the main text, a cosmologically acceptable trajectory passes
near a matter point and lands to an accelerated point. A critical point is a
good candidate for a matter point if it satisfies (i) the matter condition,
$\Omega>0$, (ii) the \textquotedblleft right\textquotedblright\ scale factor
condition, $a\sim t^{2/3}$, (or equivalently, $w_{\mathrm{{eff}}}$ close to
zero), and (iii) is a saddle point, i.e., represents a transient phase. On the
other hand, an acceptable late attractor has to be (iv) accelerated,
$w_{\mathrm{eff}}<-1/3$, and (v) stable. Points $\mathcal{B}$ and
$\mathcal{E}_{\pm}$ could be used as matter points and $\mathcal{B}%
,~\mathcal{C}_{\pm}$ and $\mathcal{E}_{\pm}$ could be used as accelerated
attractors. We are going to determine under which conditions on the parameters
$\alpha,~\beta,~\gamma$ and $Q$ there exist at the same time at least one
matter point with $w_{\mathrm{eff}}$ close to $1$, followed by at least one
accelerated future attractor.

\begin{itemize}
\item[$\mathcal{C_{\pm}}$] Following the terminology of \cite{hh}, these are
kinetic-potential scaling solutions and exist in potentials with $V_{1}>0$ for
$\alpha< \sqrt{6}$ and in potentials with $V_{1}<0$ for $\alpha> \sqrt{6}$.
They are stable and accelerated whenever
\begin{equation}
\left(  4-3 \gamma\right)  Q> \frac{2\left(  \alpha^{2} -3 \gamma\right)  }{
\alpha}~ \text{and} ~\alpha<\sqrt{2}. \label{pointC}%
\end{equation}
Hence, they are good candidates as accelerated late attractors only in
potentials with $V_{1}>0$.

\item[$\mathcal{E_{\pm}}$] Points $\mathcal{E_{\pm}}$ are
fluid-kinetic-potential scaling solutions (see also Ref. \cite{hh} for the
uncoupled case). They enter the phase space when%
\begin{equation}
\left(  4-3\gamma\right)  Q\leq\frac{2\left(  \alpha^{2}-3\gamma\right)
}{\alpha}, \label{phaseE1}%
\end{equation}
Points $\mathcal{E_{\pm}}$ may be used for the matter epoch if they satisfy
conditions (i), (ii) and (iii), i.e., if%
\begin{equation}
Q=2\alpha\frac{1-\gamma}{4-3\gamma},\gamma\leq1,~\alpha>\sqrt{\frac{3}{2}%
}\sqrt{\frac{2-\gamma}{1-\gamma}}. \label{Ematter}%
\end{equation}
In that case, points $\mathcal{E_{\pm}}$ exist only for potentials with
$V_{1}<0$. Hence, when $\mathcal{E_{\pm}}$ are used as matter points, points
$\mathcal{C_{\pm}}$ cannot be used as the accelerated attractors. The only
candidate left for the accelerated epoch is $\mathcal{B}$, but as we will see,
$\mathcal{B}$ cannot be accelerated for $Q$ given in (\ref{Ematter}). In order
for points $\mathcal{E_{\pm}}$ to be used for the accelerated epoch they have
to satisfy conditions (iv) and (v). This happens for
\begin{equation}
\left(  4-3\gamma\right)  Q<\left(  2-3\gamma\right)  \alpha~\text{and}\left(
4-3\gamma\right)  ^{2}Q^{2}-2\alpha\left(  4-3\gamma\right)  Q+6\gamma\left(
2-\gamma\right)  >0. \label{Evb}%
\end{equation}
Whenever $\mathcal{E}_{\pm}$ are accelerated attractors, the only remaining
candidate for the matter epoch is point $\mathcal{B}$, but as we shall see
right below, $\mathcal{B}$ does not satisfy the matter point conditions for
the range of the parameters given in (\ref{Evb}).

\item[$\mathcal{B}$] This is a fluid-kinetic scaled solution. Point
$\mathcal{B}$ enters the phase space when
\begin{equation}
Q\leq\sqrt{6}\frac{2-\gamma}{|4-3\gamma|}, \label{phaseB}%
\end{equation}
for $\gamma\neq4/3$ and lies always in the phase space for $\gamma=4/3$,
irrespective of the nature of the potential. For $\gamma<4/3$, condition
(\ref{phaseB}) is always satisfied for sufficiently small values of $Q$, e.g.,
$Q\lesssim1$. Matter point conditions (i), (ii) and (iii) are satisfied
whenever%
\begin{equation}
Q=\frac{\sqrt{6\left(  2-\gamma\right)  \left(  1-\gamma\right)  }}{4-3\gamma
},~\gamma\leq1,~\alpha<\sqrt{\frac{3}{2}}\sqrt{\frac{2-\gamma}{1-\gamma}}.
\label{Bmp}%
\end{equation}
On the other hand, point $\mathcal{B}$ may be an accelerated attractor if (iv)
and (v) hold, provided that (\ref{phaseB}) is satisfied. The condition for
acceleration (iv) gives%
\begin{equation}
Q<\frac{\sqrt{2\left(  2-\gamma\right)  \left(  2-3\gamma\right)  }}%
{4-3\gamma}, \label{accelB2}%
\end{equation}
with $\gamma<2/3$. Assuming (\ref{accelB2}), the stability condition, (v),
gives%
\[
\left(  4-3\gamma\right)  ^{2}Q^{2}-2\alpha\left(  4-3\gamma\right)
Q+6\gamma\left(  2-\gamma\right)  <0.\label{stabB2}%
\]
Nevertheless, $Q$ given in (\ref{Ematter}) do not satisfy (\ref{accelB2}).
Hence, matter points $\mathcal{E}_{\pm}$ cannot be combined with accelerated
attractors $\mathcal{B}$.
\end{itemize}

We conclude that there is only one case in which we have at the same time at
least one matter point and at least one accelerated attractor. This happens
whenever $\mathcal{B}$ represents the matter solution and $\mathcal{C}_{\pm}$
stand for attractors. In that case, potential has $V_{1}>0$ and the parameters
take the values%
\begin{equation}
\alpha<\sqrt{2},~\gamma\leq1,~Q=\frac{\sqrt{6\left(  2-\gamma\right)  \left(
1-\gamma\right)  }}{4-3\gamma}, \label{accept}%
\end{equation}
leading to Table 1 in the main text.

\end{document}